\documentclass[12pt]{article}
\usepackage{amsmath}
\usepackage{graphicx}
\usepackage{enumerate}
\usepackage{natbib, bbm}
\usepackage{url} 

\usepackage{amsthm}
\usepackage{lipsum}

\newtheoremstyle{break}
  {\topsep}{\topsep}%
  {\itshape}{}%
  {\bfseries}{}%
  {\newline}{}%
\theoremstyle{break}
\newtheorem{theorem}{Theorem}

\newtheoremstyle{break}
  {\topsep}{\topsep}%
  {\itshape}{}%
  {\bfseries}{}%
  {\newline}{}%
\theoremstyle{break}
\newtheorem{prop}{Proposition}

\RequirePackage[OT1]{fontenc}
\RequirePackage{amsthm,amsmath, bbm, xspace, graphics, graphicx}

\newcommand{\rmm}{\ensuremath{\mathbbm{R}^m}\xspace}

\newcommand{\tmm}{\ensuremath{\mathcal{T}_{m+3}}\xspace}

\DeclareMathOperator*{\argmin}{arg\,min}

\newcommand{\blind}{1}

\begin{document}

\def\spacingset#1{\renewcommand{\baselinestretch}%
{#1}\small\normalsize} \spacingset{1}


\if1\blind
{
  \title{\bf Confidence sets for phylogenetic trees}
  \author{Amy Willis \hspace{.2cm}\\
    Department of Biostatistics, University of Washington}
  \maketitle
} \fi

\if0\blind
{
  \bigskip
  \bigskip
  \bigskip
  \begin{center}
    {\LARGE\bf Confidence sets for phylogenetic trees}
\end{center}
  \medskip
} \fi

\bigskip
\begin{abstract}
Inferring evolutionary histories (phylogenetic trees) has important applications in biology, criminology and public health. However, phylogenetic trees are complex mathematical objects that reside in a non-Euclidean space, which complicates their analysis. While our mathematical, algorithmic, and probabilistic understanding of phylogenies in their metric space is mature, rigorous inferential infrastructure is as yet undeveloped. In this manuscript we unify recent computational and probabilistic advances to construct tree--valued confidence sets. The procedure accounts for both centre and multiple directions of tree--valued variability. We draw on block replicates to improve testing, identifying the best supported most recent ancestor of the Zika virus, and formally testing the hypothesis that a Floridian dentist with AIDS infected two of his patients with HIV. The method illustrates connections between variability in Euclidean and tree space, opening phylogenetic tree analysis to techniques available in the multivariate Euclidean setting.
\end{abstract}

\noindent%
\vfill

\newpage
\spacingset{1.45} 

\section{Introduction} \label{intro}

Evolutionary histories are key data objects in biology, biogeography, criminology, anthropology and immunology. In addition to illuminating interesting ancestral connections, their careful analysis has aided in freeing the innocent \citep{scad}, reducing accidental disease transmission by healthcare providers \citep{ou92}, and identifying perpetrators of willful HIV infection \citep{hill94}. Broader applications in the medical sciences, such as modeling brain and lung networks \citep{amenta2015quantification, bendich2016persistent,feragen2012hierarchical,skwerer2014tree}, further motivate interest.

The development of models for inferring evolutionary histories, or {\it phylogenies}, has become highly sophisticated since its genesis \citep{cavalli, felsenstein:1983vc}. Different models may account for such varied biological possibilities as stochastic coalescence \citep{heledd},  gene duplication  \citep{rasmussen2012unified}, and hybridization \citep{gerard2011estimating}. However, it may not be the case that there is a single evolutionary history unanimously implied by all genetic loci, and different genetic sites may conflict with respect to the implied phylogeny \citep{reid}. Between the 1960's and the early 2000's this line of research focused on {\it consensus} methods, which unify collections of evolutionary histories into a single tree.
However, rather than summarizing the collection by a single tree, it is now common to analyze the collection of trees. There are many modern methods that utilize the elegant mathematics of tree space to propose new analysis tools that generate new insights \citep{bbb, nye11,kdetrees}.

Despite great gains made with respect to exploring (a) the center of a collection of trees \citep{bbb}, and (b) the directions of their variability \citep{nye11,Nye:2014eh, nye2016principal}, the literature lacks a statistical method that simultaneously considers both these issues \citep{owen17}.
The approach presented in this paper quantifies multiple directions of variability along with centre, offering a new perspective on  the relationship between tree space and $\mathbbm{R}^n$ and providing a usable solution to the important open problem of constructing confidence sets for phylogenetic trees \citep{holmes2003bootstrapping,Holmes:2005wu,lubiw2016shortest}.

We begin with an overview of existing mathematical infrastructure, including  tree space, manifolds and homeomorphisms,  and central limit theorems (Section~\ref{infrastructure}), before developing the necessary statistical infrastructure (Section~\ref{method}). We describe our confidence set construction procedure in detail (Section~\ref{procedure}) and investigate coverage (Section~\ref{coverage}), then demonstrate its utility for detecting splits of weak and strong support and in tree-valued hypothesis testing (Section~\ref{examples}). Among our examples we investigate the biogeography of the Zika virus as well as a forensics investigation. A discussion of  experimental design for phylogenetics and other questions raised by this method, as well as  the relationship between statistics on tree space and Euclidean space,  concludes the paper (Sections~\ref{discussion} and \ref{conc}).

\section{Structure} \label{infrastructure}

The key innovation of this paper is the construction of confidence sets for phylogenetic tree-valued parameters. We progress the major advances presented by \cite{blow2} into a complete statistical framework for inference. To do this we rely on a number of mathematical constructions, including tree space, Fr\'{e}chet means, and the tree-log-map. While this section is not intended to be self-contained, and we refer the reader to the references for more details, we briefly review some necessary concepts and introduce notation and a new concept of tree set measurability.

\subsection{Tree space}
The metric space of phylogenetic trees $(\tmm, d)$, or {\it tree space}, is a complete separable metric space \citep{bhv, Benner:2014uy,cptmle} that permits comparison between phylogenetic trees with the same leaf set of cardinality $m+3$. The distance $d(T_i, T_j)$ between two trees $T_i$ and $T_j$ accounts for differences with respect to both their tree topologies (branching structure) and branch lengths. The space is constructed by representing each of the $(2m+1)!!$ possible tree topologies by a single non-negative Euclidean orthant of dimension $m$ (the largest  possible number of internal branches). The orthants are then ``glued together'' \cite[p. 12]{bhv} along appropriate boundaries. Specifically, nearest neighbor interchange (NNI) topologies lie in adjacent non-negative orthants along the boundary corresponding to the collapse of the relevant NNI edge. Orthants and orthant boundaries are also called {\it strata}, and trees with the largest possible number of internal branches are said to fall in {\it top-dimensional strata} while trees with less than the largest possible number of internal branches fall in {\it co-dimensional strata}. Figure \ref{tree_space_illustration} shows the structure of tree space with 5 leaves, $\mathcal{T}_5$,  around a single co-dimension 1 stratum, with the 3 associated NNI topologies. Construction of the space is due to \cite{bhv}.\footnote{For generality we consider trees to be unrooted, though by designating a particular leaf as the root the restriction to rooted trees is trivial. Note that \cite{blow2} considered rooted trees with $m$ internal edges;  this difference accounts for our use of \tmm, in contrast to their use of $\mathcal{T}_{m+2}$.}

\begin{center}
\begin{figure}
\includegraphics[trim = 3cm 3cm 3cm 3cm, scale=0.7]{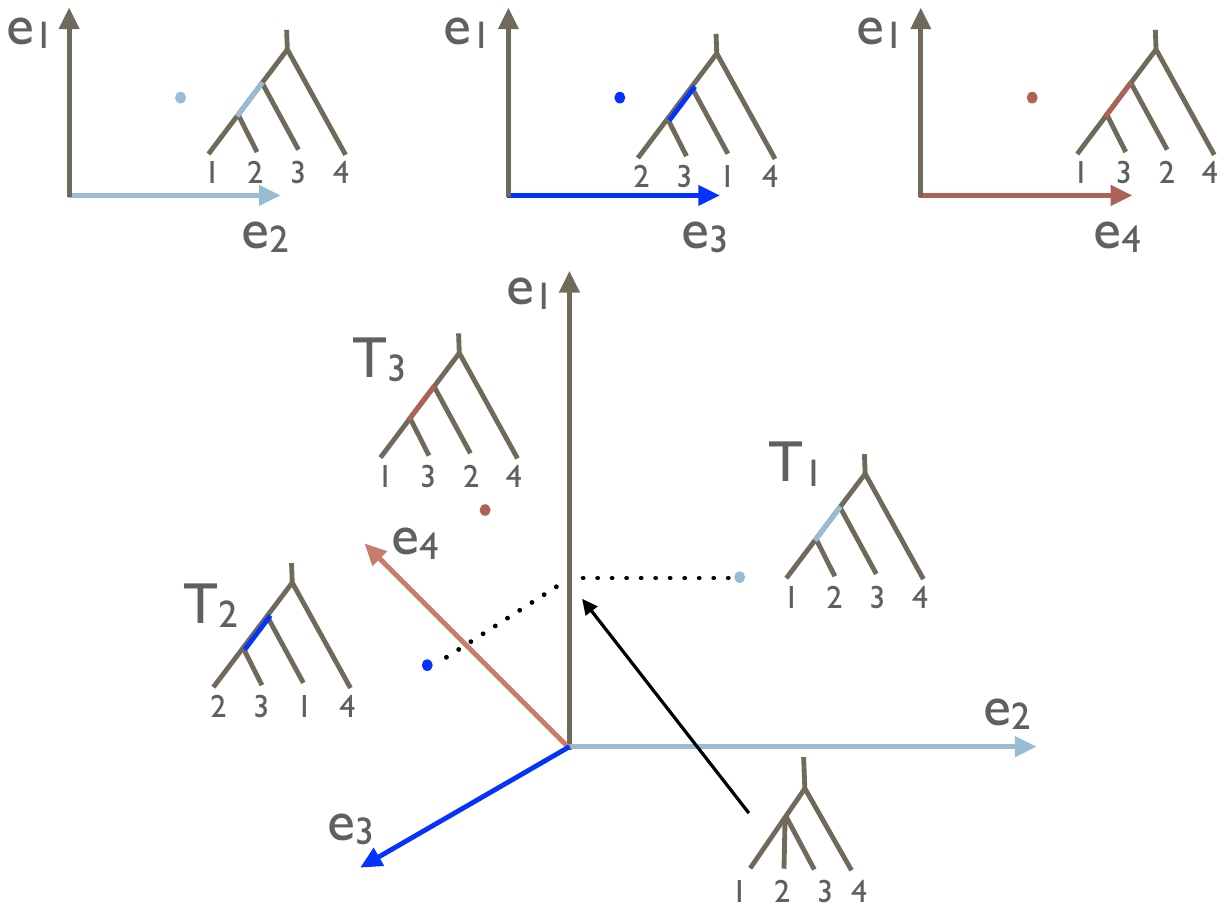}
\caption{The structure of tree space with 5 leaves, $\mathcal{T}_5$,  around a single co-dimension 1 stratum. Trees $T_1$, $T_2$ and $T_3$ mutually differ by a nearest neighbor interchange (NNI) move, and hence the orthants (or top-dimensional strata) associated with their topologies are connected along the co-dimension 1 stratum. The dotted line between $T_1$ and $T_2$ is the unique shortest path between these trees, and the color coding refers to the branches and branch lengths of the trees. }
\label{tree_space_illustration}
\end{figure}
\end{center}

The space is {\it nonpositively curved} \citep{bhv, Bridson:1999ky}, which has resulted in  algorithms for calculating geodesic paths \citep{owenfast}, means \citep{bacak,mopr} and principal paths \citep{Nye:2014eh}. Furthermore, non-positive curvature (NPC) of the space results in unique {\it Fr\'{e}chet means}: for a collection of trees $T_1, \ldots, T_n \in \tmm$, the sample Fr\'{e}chet mean
\begin{align}
\hat{T}_n= \hat{T}_n(T_1,\ldots,T_n) := \argmin_{t \in \tmm} \sum_{i=1}^n d(T_i, t)^2 \label{smean}
\end{align}
is guaranteed to be unique \citep{sturm}. 

\subsection{Probability triples}
Given guarantees of tree mean uniqueness, we turn our attention to their inference. To maintain  rigor, a number of generalizations need to be specified. In the context of tree-valued Brownian motion, \cite{nye2015convergence} explicitly constructs a $\sigma$-algebra on the metric space of paths in tree space generated by the $\ell_\infty$  metric. However, here we require the algebra on the space of trees itself. 

We begin by constructing a $\sigma$-algebra on $\tmm$. Because \tmm is a metric space we have a well-defined concept of open balls, and  can generate the Borel algebra by their countable unions, intersections and relative compliments. By construction, this is a $\sigma$-algebra.

Enumerate the tree topologies $j = 1,\ldots, (2 m+1)!!$. With our  $\sigma$-algebra  we can now define a measure on this algebra by writing any set  in \tmm in the form $A = A_0 \cup\left(\bigcup_{j=1}^{(2m+1)!!} A_j\right)$ where $A_j$ is a collection of trees with the $j$-th topology, and $A_0$ is a collection of trees  with one or more internal vertices of degree 4 or greater  (equivalently, trees on the orthant boundaries of \tmm). We then define $\nu(A) = \sum_{j= 1}^{(2 m+1)!!}\nu_B(A_j)$ for $\nu_B$ the Euclidean Borel measure of dimension $m$. This preserves $\sigma$-additivity because the Borel measure of any orthant boundary is zero, and there are only finitely many such boundaries. This construction is not a complete measure space, so we append all sets of  measure zero to give an analogue of Lebesgue-measurable sets in \tmm, which we call $\mathcal{L}(\tmm)$. The latter construction will be implicitly  used as our $\sigma$-algebra henceforth, and we have a  complete measure space $(\Omega, \mathcal{L}(\tmm), \nu)$. Finally, for a probability measure $F: \mathcal{L}(\tmm) \rightarrow [0,1]$, defined with respect to the volume measure $\nu$, we obtain a  probability triple  $(\Omega, \mathcal{L}(\tmm), F)$.


\subsection{Limit theorems and the log map}
While theory for central limit theorems for Fr\'{e}chet means on metric spaces has been well-developed \citep{Bhattacharya:vk}, these generally rely on homeomorphisms to $\mathbbm{R}^n$ from measurable subsets of the space known to contain the true Fr\'{e}chet mean $\mu$ of the probability measure $F$, $$\mu = \argmin_{u \in \tmm} \int d(q, u)^2 F(dq).$$ However, due to the stratified structure of \tmm, inverse functions will not exist for candidate homeomorphisms except restricted to subsets wholly contained in a single orthant. Thus without assuming the topology of the true mean {\it a priori}, general results for CLTs  on manifolds are insufficient for tree mean inference.

To overcome these difficulties, \cite{blow2} developed a mapping from \tmm to \rmm and proved mapped multivariate normality of the sample Fr\'{e}chet mean around the true Fr\'{e}chet mean, deriving expressions for the covariance based on the distribution $F$ (see also \cite{barden2017}). While the proposed mapping, called the log map, is not a homeomorphism, it elegantly deals with samples from multiple topologies, taking advantage of similarities between tree space and Euclidean space while respecting their different combinatorial structures.

The log map function of \cite{blow2}, $\log_{T^*}(T)$, captures both the distance and direction from a base tree $T^*$ to a target tree $T$. For now, we consider only base trees off the orthant boundaries, though we return to this issue in Sections  \ref{turtles} and \ref{discussion}. 
Formally, $\log_{T^*}(T):  \tmm \rightarrow \rmm$ is defined as $$\log_{T^*}(T)=d(T ^*,T)\mathbf{v}_{T^*}(T),$$ where $d(T^*,T)$ is the geodesic distance between $T^*$ and $T$, and $\mathbf{v}_{T^*}(T)$ is a specifically chosen unit vector from $T^*$ to $T$ that reflects the direction of the first segment of the geodesic (details below). The modified log map (MLM) is a translation of the log map. It positions this vector to originate from the base tree, $$\Phi_{T^*}(T) = \log_{T ^*}(T)+\mathbf{t}^*,$$ for $\mathbf{t}^*$ the coordinates in \rmm  of  $T ^*$'s internal edge lengths. Note that while $T^*$ is a phylogenetic tree, $\mathbf{t}^*$ is a vector. For a description of the correct permutations of the ordering of the edge lengths necessary to maintain invariance across multiple argument topologies, see \cite{blow2}.

\begin{center}
\begin{figure}
\includegraphics[trim = 6cm 0cm 0cm 0cm, scale=0.7]{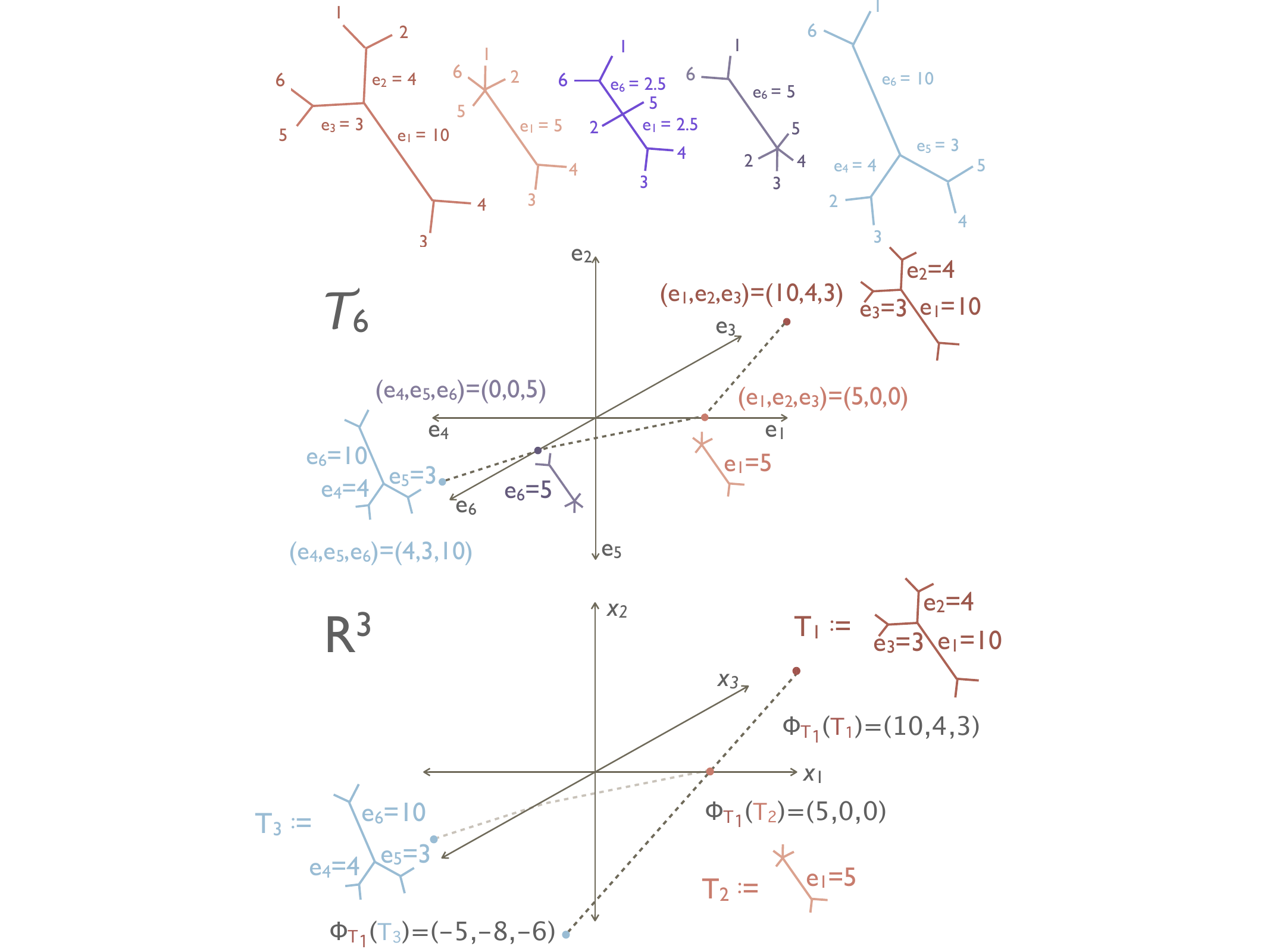}
\caption{The geodesic path between 2 trees  with 6 leaves (top panel), a representation of this path in $\mathcal{T}_6$ (middle panel), and the modified log map (MLM) with respect to the  tree $T_1$  (bottom panel). The distance between $T_1$ and $T_3$ is $15\sqrt{2}$, hence  $\Phi_{T_1}(T_3)=(10,4,3)+15 \sqrt{2} \frac{(5 , 0 , 0)-(10, 4, 3)}{||(5 , 0 , 0)-(10, 4, 3)||}=(-5, -8, -6).$ 
The log map captures both the length and direction of the geodesic by extending the first linear segment of the geodesic path.}
\label{logmap_illustration}
\end{figure}
\end{center}

The intuition behind $\mathbf{v}_{T^*}(T)$ can best be illustrated via the MLM, and may be seen in Figure~\ref{logmap_illustration}. For a target tree in the same orthant as the base tree (identical topologies), the MLM coincides with a Euclidean representation of the target tree, that is, is an $m$-vector with all positive components reflecting the lengths of the target tree's internal branches. For a target tree in an adjacent orthant (nearest neighbor interchange topology), the MLM vector has a single negative component with magnitude equal to the length of the branch present on the target tree but not present on the base tree, with the remaining components positive (adjusted to reflect the branch lengths of the target tree). If the target tree is more topologically distinct than a NNI interchange from the base tree, the simplest visualization is to follow the initial segment of the geodesic path (the segment contained in the same orthant as the base tree) for the length of the geodesic across (potentially more than one) Euclidean orthant boundaries (Figure~\ref{logmap_illustration}).
Formally, denote the support of the geodesic (the sequence of orthants that the geodesic path traverses; see \cite{owenfast}) between $T^*$ and $T$ by $\mathcal{A}=(A_1, \ldots, A_k)$ and $\mathcal{B}=(B_1, \ldots, B_k)$, and choose $\lambda$ to be any positive number strictly less than $ \frac{||A_1||}{||A_1||+||B_1||}$. Let $S^*_\lambda$ be the tree that is of fraction $\lambda$ along the geodesic path between $T^*$ and $T$, noting that, by construction, $S^*_\lambda$ has the same topology as $T^*$ \citep{owenfast}. Then
$$\Phi_{T^*}(T) = \mathbf{t}^* + \frac{\mathbf{s}^*_\lambda - \mathbf{t}^*}{|\mathbf{s}^*_\lambda - \mathbf{t}^*|} d (T^*, T),$$
where $\mathbf{s}^*_\lambda$ and $\mathbf{t}^*$ are the coordinates in $\mathbbm{R}^m$ of the edge lengths of $S^*_\lambda$ and $T^*$. Note that the particular choice of $\lambda$  does not affect $\Phi_{T^*}(T)$ because  $\frac{\mathbf{s}^*_\lambda - \mathbf{t}^*}{|\mathbf{s}^*_\lambda - \mathbf{t}^*|}$ is a unit vector.

It is important to note that the MLM  is not a bijection. If any coordinate in the MLM is negative, the argument tree cannot be uniquely recovered. In this way, the MLM does not preserve topological information. A detailed illustration of this issue is given in the Supplementary Appendix.


\section{Statistical infrastructure} \label{method}

Many modern investigations in phylogenetics give rise to tree-valued information, where each of several sources or combinations of data suggest their own, possibly conflicting evolutionary histories. Thus, as the primary data objects for our method, consider $n$ unrooted trees sampled on the same collection of $m+3$ taxa, $m \geq 1$. Thus we have $T_1, T_2, \ldots, T_n \in \tmm$ as our observations.

We consider these observations to be observed {\it iid} from some unknown distribution $F$ on \tmm which represents the variability in the set of  trees (which, if desired, could be built up through a stochastic model for nucleic acid base substitutions or a population dynamics model).  
We now state the key result on which our method is based.

\begin{theorem} \label{central}

Let $T_1, \ldots, T_n \in \tmm$ be observed {\it iid} from some distribution $F$. Suppose that $F$ has a finite Fr\'{e}chet function $G(u) = \int_{\tmm} d(q, u)^2 F(dq)$, and Fr\'{e}chet mean $T^* = \argmin_{u \in \tmm} G(u)$ in a top-dimensional stratum. Furthermore, suppose $F$ satisfies $F(D_{T^*})=0$, where $D_{T^*}$ is the set of trees that lie on the boundaries of the maximal cells of $T^*$ (\citet[Definition 2]{blow2}). Then, for the sample Fr\'{e}chet mean $ \hat{T}_n$, as $n$ becomes large,
$$\sqrt{n}\left( \hat{T}_n - T^*\right) \overset{\mathcal{D}}{\rightarrow} \mathcal{N}\left(0, \Sigma\right),$$
for $\Sigma$ a covariance matrix determined by the structure of $F$.
\end{theorem}

This result is due to \citet[Theorem 2]{blow2}.  Here we use  $\Sigma$ to denote the covariance of the distribution in order to emphasize that this matrix is estimable (discussed in Section \ref{covariance_estimation}).  The notation $\left( \hat{T}_n - T^*\right)$ refers to $\left(\Phi_{\hat{T}_n}(\hat{T}_n)-\Phi_{T^*}(T^*)\right)$ subject to the ordering of the components of the vector-valued function $\Phi_{\hat{T}_n}(\cdot)$ matching the order of the components of the vector-valued function $\Phi_{T^*}(\cdot)$. Note that the proof of the theorem by \cite{blow2} accounts for uncertainty in estimating both $T^*$ and $\Phi_{T^*}(\cdot)$.

This theorem forms the foundation of our tree inference procedure. However, we must discuss estimation of $T^*$, $\Phi_{T^*}(\cdot)$ and $\Sigma$ in order to bridge the gap between the existing probability theory and statistics. We address these in turn below.

\subsection{Mean tree}

Given Theorem \ref{central}, a natural candidate for estimating $T^*$ is $\hat{T}_n$, the sample mean (see Eq.~(\ref{smean})). Furthermore, \citet[Theorem p.~592]{Ziezold:wo} guarantees strong \tmm--consistency of the (unique) sample Fr\'{e}chet mean  under weaker conditions than \citet[Theorem 2]{blow2} and Theorem \ref{central}. Efficient algorithms for sample Fr\'{e}chet mean computation lend an additional appealing quality \citep{bacak, mopr}, as do the results of \cite{bbb}.

As statistical inference for tree space develops, estimators of $T^*$  other than $\hat{T}_n$ are likely to be developed, supported by their own CLTs and LLNs. The procedure described in Section~\ref{procedure} would be unchanged  if the practitioner were to use a different normally distributed  true Fr\'{e}chet mean estimate.

\subsection{Log-map function}

Masked by notation is the need to estimate $\Phi_{T^*}(\cdot)$, the log-map function. This function is fully determined by its base tree $T^*$,  which must be estimated. The above discussion regarding tree mean estimation points to estimating the function $\log_{T^*}(\cdot)$ by $\log_{\hat{T}_n}(\cdot)$, however, functional consistency needs to be verified. Since \tmm is not Hilbert, and the function under question is not linear, we do not know of any results that immediately guarantee this.
\begin{prop}
Consider trees $T_1, \ldots, T_n \in \tmm$ drawn from some distribution $F$ with true Fr\'{e}chet mean $T^*$, not located on an orthant boundary,  that satisfies
\begin{align}
\int_{\tmm} d (\mathbf{0}_{\tmm}, T)^2\ dF(T) < \infty, \label{condition}
\end{align}
for $\mathbf{0}_{\tmm}$ the origin (star tree) in \tmm.
For a fixed tree $t \in \tmm$, consider the function $\log_{(.)}(t):\tmm \rightarrow \rmm$. Then for all $t \in \tmm$ we have that $\Phi_{\hat{T}_n}(t)$ converges almost surely to $\Phi_{T^*}(t)$ in $\rmm$ as $n \rightarrow \infty$.
\end{prop}
\begin{proof}
Almost sure convergence is preserved by addition and multiplication,  thus it is sufficient to show that (a) $ \hat{\mathbf{t}}_n^* \overset{a.s.}{\rightarrow}  \mathbf{t}^*$ in \rmm, (b) $d(\hat{T}_n, t) \overset{a.s.}{\rightarrow} d(T^*, t)$ in $\mathbbm{R}$, and (c) $\mathbf{v}_{\hat{T}_n}(t) \overset{a.s.}{\rightarrow} \mathbf{v}_{T^*}(t)$ in $\mathbbm{R}$ for all $t$.
From \cite{Ziezold:wo}, combined with the assumption that $\int_{\tmm} d (\mathbf{0}_{\tmm}, T)^2\ dF(T)$ $< \infty$, we have that $\hat{T}_n \overset{a.s.}{\rightarrow} T^*$ in \tmm. Hence for $n>m$ for some $m$,  we are guaranteed that the sample mean will be in the same orthant as the true mean and thus for $n > m$, $|\hat{\mathbf{t}}_n^* - \mathbf{t}^*| = d(\hat{T}_n, T^*) \overset{a.s.}{\rightarrow} 0$, hence (a). $d(\cdot, \cdot)$ is a bicontinuous function (metrics are guaranteed this property), and thus the continuous mapping theorem guarantees (b). For (c),  we claim $\mathbf{v}_{(\cdot)}(t)$  is a continuous function. We need only consider possible discontinuities in the carrying orthant sequence ($d(\cdot, \cdot)$ may be bicontinuous but this is no guarantee that the path will be). However, as noted by
\cite{blow2},
the polyhedral subdivision varies continuously with respect to the base tree, and directional derivatives of the log map are well-defined. Hence the carrying orthant sequence is continuous, and (c) follows.
\end{proof}

Thus we estimate the true modified log map function with the sample modified log map function.

\subsection{Covariance estimation}  \label{covariance_estimation}

Perhaps the most substantial complication when transitioning from Theorem \ref{central} to a statistical inference procedure is the estimation of the covariance matrix $\Sigma$. While this matrix can (theoretically) be calculated given known $F$, we wish to avoid strong assumptions on $F$, and in particular, we wish to avoid {\it structural} assumptions. Furthermore, calculating this matrix in practice would require integration across multiple orthants, and despite the advances in tree integration by \cite{Weyenberg:2016ti} and \cite{weyenberg2016normalizing}, this will generally be intractable for distributions with full support on the space.

The approach we present here takes advantage of the intrinsic connection between tree space and Euclidean space. Consider our estimate of $\Phi_{T^*}(\cdot)$, $\Phi_{\hat{T}_n}(\cdot)$, and consider the collection of $\mathbbm{R}^m$-valued objects $\Phi_{\hat{T}_n}(T_1),$ $\Phi_{\hat{T}_n}(T_2),$ $\ldots,$  $\Phi_{\hat{T}_n}(T_n)$, which we know has sample mean $\Phi_{\hat{T}_n}(\hat{T}_n)$ \cite[Lemma 3, setting $\mu(T)=\mathbbm{1}_{\{T\in \{T_i\}\}}(T)$ and the base tree as $\hat{T}_n$]{blow2}. Because $\Phi_{\hat{T}_n}(\cdot)$ collapses the structure of tree space around $\hat{T}_n$,  the vectors $\Phi_{\hat{T}_n}(T_1),$ $\Phi_{\hat{T}_n}(T_2),$ $\ldots,$  $\Phi_{\hat{T}_n}(T_n)$ give a collection of $n$ $\mathbbm{R}^m$-valued observations that form a point cloud around their mean.

We propose to use an unstructured covariance estimator to estimate $\Sigma$ via the  covariance of $\Phi_{\hat{T}_n}(T_1),$ $\Phi_{\hat{T}_n}(T_2),$ $\ldots,$  $\Phi_{\hat{T}_n}(T_n)$. This does not require knowledge of $F$  nor calculation of integrals over tree space. If a  large tree sample is available, the classic estimator
\begin{align} \label{covariance_unstructured}
S=\frac{1}{n-1} \sum_{i=1}^n \left(\Phi_{\hat{T}_n}(T_i) - \Phi_{\hat{T}_n}(\hat{T}_n)\right) \left( \Phi_{\hat{T}_n}(T_i) - \Phi_{\hat{T}_n}(\hat{T}_n) \right)^T
\end{align}
is unbiased for $\Sigma$ and has distribution that converges to a rescaled Wishart distribution if the $\Phi_{\hat{T}_n}(T_i)$'s are approximately normally distributed \citep{tim_multivariate} (see Section \ref{procedure} for a discussion). It is important to note that $S$ will be  unstable when the number of trees is small compared to the number of leaves on the trees, and stability could be introduced by imposing sparsity on the estimate. 

\section{Procedure} \label{procedure}

We now give a description of our proposed confidence set procedure and discuss its underpinning assumptions.

\begin{enumerate}
\item For trees $T_1, \ldots, T_n \in \tmm$, calculate the mean tree $\hat{T}_n$ using the proximal point algorithm of \citet[ Algorithm 4.2]{bacak} (see also \cite{Benner:2014uy}).
\item Calculate the Euclidean projections  under the sample MLM: $\Phi_{\hat{T}_n}(T_1),$ $\Phi_{\hat{T}_n}(T_2),$ $\ldots,$  $\Phi_{\hat{T}_n}(T_n)$.
\item Estimate the true tree $T^*$ by $\hat{T}_n$, and the precision matrix $\Sigma^{-1}$  by $S^{-1}$ (Eq.~\ref{covariance_unstructured}) or with a sparse estimate.
\item Define
\begin{align}
\overline{\Phi_{\hat{T}_n}(T)}  = \frac{1}{n}\sum_{i=1}^n {\Phi_{\hat{T}_n}(T_i)}  \notag,
\end{align}
an object in \rmm.
Construct a $100(1-\alpha)$\% confidence hull for $T^*$, the true Fr\'{e}chet mean of the data generating distribution, via
\begin{align}
A= \bigg\{ T_0 &\in \tmm: \label{set} \\ \notag
&\left(\overline{\Phi_{\hat{T}_n}(T)} - \Phi_{\hat{T}_n}(T_0) \right)^T S^{-1} \left(\overline{\Phi_{\hat{T}_n}(T)} - \Phi_{\hat{T}_n}(T_0) \right)< \frac{m(n-1)}{n (n-m)} F_{m, n-m} (1-\alpha)  \bigg\}, \notag
\end{align}
modifying the pivot distribution or degrees of freedom as appropriate for a different estimator of the precision. The above combines Euclidean multivariate results from \cite{tim_multivariate} with the tree CLT of \cite{blow2}.
\item For  a candidate true Fr\'{e}chet mean tree $T_0$, if $T_0 \in A,$ conclude that $T_0$ is contained in the $100(1-\alpha)$\% confidence set for the true tree mean $T^*$.
\end{enumerate}


It is important to note that the confidence procedure is limited by the disagreement between the distribution of the $\Phi_{\hat{T}_n}(T_i)$'s and the multivariate normal distribution. The extent of this disagreement will depend on the true distribution of the trees, $F$, which is unknown necessarily. However, the coverage simulations of Section \ref{coverage}  suggest that the disagreement is not too severe, at least for low parameter $F$'s. Investigations of more complex cases is an ongoing subject of research. Furthermore, if normality is implausible but the more flexible assumption of ellipticity would suffice, \citet[Theorem 2.1]{sutradhar1989generalization} could be used to derive more appropriate, heavier-tailed asymptotics. Bootstrapping from the sampling distribution (Eq.~\ref{set}) could also be employed, though the trade-off between a weakly violated assumption and failing to adjust for out-of-sample variability \citep{holmes2003bootstrapping} is not clear.



\section{Coverage} \label{coverage}

We explore  coverage  for two different phylogenetic trees and a variety of sample sizes. The investigation was conducted by first selecting a tree $\tau$ and using this tree to simulate $1000 \times n$ draws of 350 aligned base pairs under a simple HKY model using seq-gen \citep{rambaut1997seq}, then estimating the $1000 \times n$ trees from the aligned base pairs under a HKY model using phyML \citep{Guindon:2003ba}. These trees were then grouped into 1000 collections of $n$ samples, and for each collection the sample mean was calculated and the confidence set from Eq.~\ref{set} was constructed for each level $\alpha$  of interest.
 The proportion of the 1000 confidence sets that contain the true  Fr\'echet  tree gives the estimated coverage of the $100(1-\alpha)$\% confidence set.
Noting that $\tau$ may not be equal to the true Fr\'echet mean, we approximated the true mean by simulating 100,000  trees  and calculating their sample mean using the proximal point algorithm.

The results of the simulation study are reported in Table \ref{coverage_table} for $\tau$ as the Zika tree (6 leaves, see Figure \ref{zika_figure}) and the HIV tree (5 leaves, see Figure \ref{hiv_figure}). The observed coverage is always close to nominal, and does not deviate by more than 1.3\%. Sample size does not have a consistent effect on coverage, though in general larger sample sizes increase coverage. Additional coverage simulations investigating the effect of longer sequences, different substitution models,  shorter branch lengths, and larger trees are available in the Supplementary Appendix, where we find that none of these factors generally affect coverage, which is consistently close to nominal.

We also report the percentage of sample means that are concordant (i.e. agree in topology) with the true Fr\'echet mean tree, and the percentage of trees that are concordant with the true Fr\'echet mean in Table 1 as ``($\bar{T}_n$, $T_i$) Concordance''. Unsurprisingly (given the results of \cite{Ziezold:wo}), the proportion of Fr\'echet means that agree with the true tree increases with sample size, and this proportion is much higher than the percentage of trees that agree with the true tree. This indicates that Fr\'echet mean trees are superior to individual trees in estimating the true Fr\'echet mean topology.

\begin{table*}
\caption{Estimated coverage of the confidence set procedure when sequence data is generated by an HKY process. 
The percentage of confidence sets containing the  Fr\'echet mean of the data generating process is reported as Coverage.  Exact coverage would be given by $(90, 95, 99)$. The percentage of sample means that agree in topology with $\tau$ is reported as $\bar{T}_n$ concordance, and the proportion of estimated trees from simulated sequences that agree in topology with $\tau$ is reported as $T_i$ concordance.}
\label{coverage_table}
\begin{tabular}{c|c|c|c}
\hline
$\tau$: True tree  & $n$ & ($\bar{T}_n$, $T_i$) Concordance & Coverage: $\alpha = (0.10, 0.05, 0.01)$ \\
\hline
HIV tree  & 20 &  (98.3, 67.7)  &$(89.0, 93.9, 98.4)$ \\ 
HIV tree  &50 & (99.6, 68.7)   &$(90.0, 95.3, 99.2)$  \\
HIV   tree  & 100  & (99.9, 68.9) &  $(91.2, 95.4, 98.7)$ \\
\hline
Zika tree  &20 &  (98.3, 66.8)  &$(89.8, 94.2, 98.7)$ \\ 
Zika tree  & 50 & (99.6, 67.3)   &$(91.1, 96.1, 98.9)$  \\
Zika   tree  & 100 & (99.9, 67.3)  &  $(89.6,  94.2,  98.6)$ \\%
\hline
\end{tabular}
\end{table*}

\section{Examples} \label{examples}

We now demonstrate that the method gives interesting new analyses in three datasets.

\subsection{Zika origins}
The implications of the Zika virus' spread has caught worldwide attention. The virus is known to have originated in Africa, with media releases in South America purporting that the virus arrived across the Atlantic ocean \citep{zika_article_2,zika_article_1}. We investigate this claim by tracing the biogeography of the current Zika outbreak in South America, concluding strong evidence that the most recent ancestor of the South American outbreak was in fact from the Pacific, whose most recent ancestor was from South-Eastern Asia, and thus that the virus traveled east from Africa, rather than west.

All available complete Zika genome sequences with complete location and year information were obtained from GenBank on June 7, 2016. We categorized the sequences by location and year, resulting in 6 categories (leaves), which are shown in Figure \ref{zika_figure} (left panel). We considered different samples within the same category as block replicates, and drew one sample from each category, aligned the sequences,  and fit a  HKY model to the phylogeny. We repeat this 108 times to have 108 evolutionary histories reflecting the within-virus variability.

Figure \ref{zika_figure} (left panel) shows the sample Fr\'{e}chet mean of the 108 trees. A branch separating recent South American strains and recent Pacific strains is present on the sample mean tree. However, the sample mean tree alone is insufficient to assess if this branch is present on the true mean tree, and for this we employ the proposed confidence procedure.

The MLMs of the 108 trees (relative to the sample Fr\'echet mean tree) are shown in Figure \ref{zika_figure} (right panel), along with the 99.9\% confidence set for the MLM of the true Fr\'echet mean tree. The confidence set for the log mapped tree does not contain any vectors with negative coordinates. Equivalently, the confidence set for the true tree only contains trees with the topology shown in Figure \ref{zika_figure}. In particular, all trees in the confidence set contain a branch that separates the South American and recent Pacific strains of the virus (branch 1). We therefore conclude that the virus travelled to South America via the Pacific, rather than descending from an African strain, corroborating recent results \citep{weaver2016zika,wang2016mosquitos,shen2016phylogenetic}.

\begin{center}
\begin{figure}
\includegraphics[trim = 0.5cm 12cm 0cm 1cm, scale=0.55]{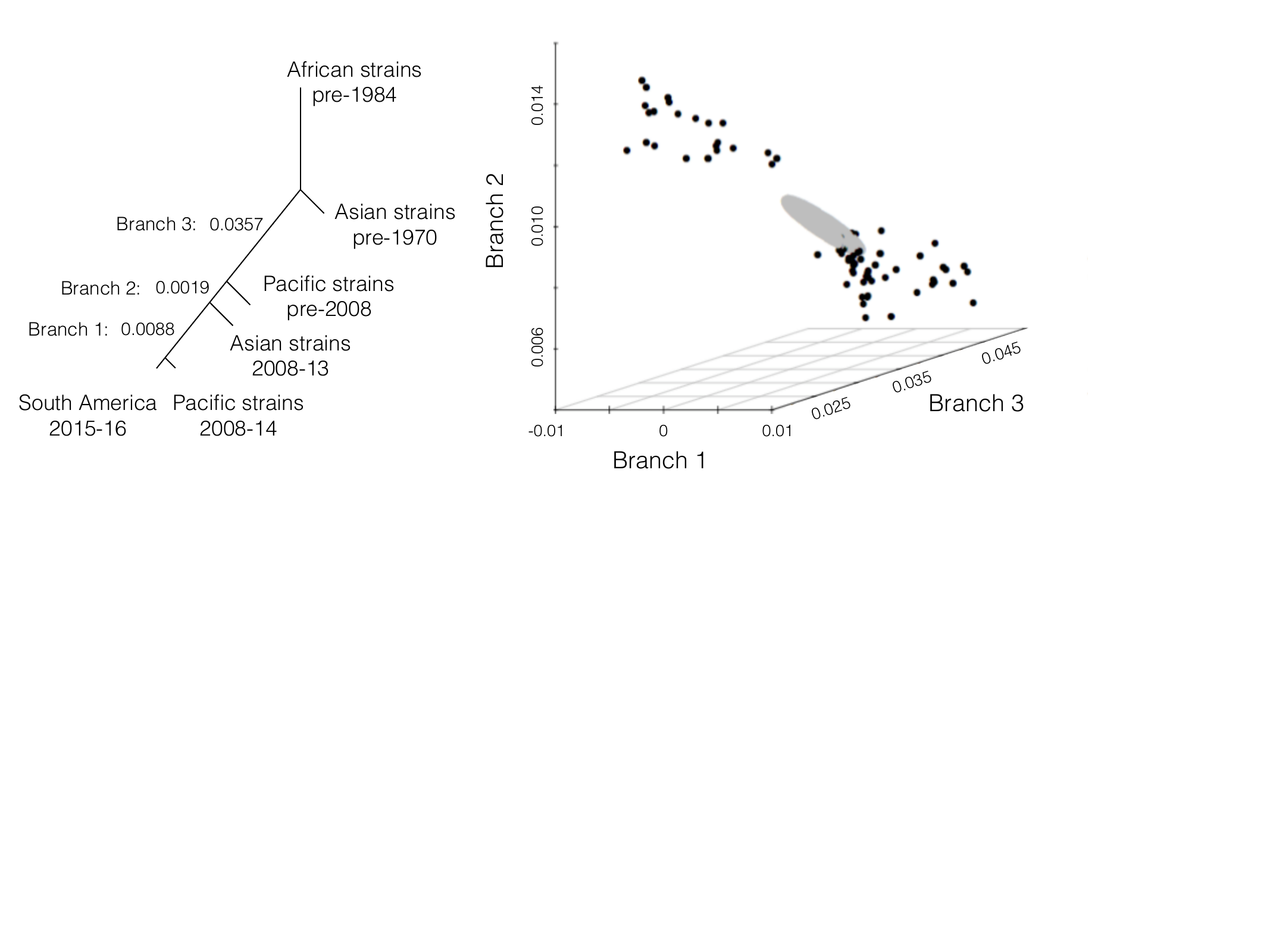}
\caption{(left) The  Fr\'echet  mean of 108 Zika phylogenies. Branch length units are substitutions per site. (right) The MLMs of the 108 phylogenies with respect to the sample  Fr\'echet mean (black points) and the 99.9\% confidence set for the MLM of the true Fr\'echet mean (grey ellipsoid). }
\label{zika_figure}
\end{figure}
\end{center}

There are 2 topologies present in the trees $T_1, \ldots, T_{108}$ (shown in the Supplementary Appendix). Due to the stratified structure of tree space, it is not possible for a confidence set constructed according to the procedure of Section \ref{procedure} to contain 2 topologies, even if only 2 topologies are represented in the data. This is because any confidence set that crosses an orthant boundary must contain all tree topologies that share that boundary (of which there are 3 or more, see Figure 1). Thus the confidence set may not exactly reflect the topologies present in the sample.

\subsection{HIV forensics}

In this example we investigate the hypothesis that two HIV-positive patients of a Floridian dentist with AIDS contracted HIV from the dentist. This question was formally investigated by the National Centre for Infectious Diseases  in 1992, culminating in a report concluding transmission to the patients from the dentist \citep{ou92}. The report considered several different analyses, including the within-patient HIV variation (HIV is known to mutate rapidly), and a phylogenetic analysis. Here we unify these two  distinct elements of the report into a single inferential method that accounts for both within- and across-patient variation of the virus.

Amino acid sequences from the V3  region of the HIV virus of the  dentist (8 replicates), patient A (6 replicates), patient B (14), a local control (2), and a non-local control (2) were obtained from GenBank. 
There are $8 \times 6 \times 14 \times 2 \times 2 = 2688$ ways to choose a single dentist sequence, patient A sequence, patient B sequence, local control sequence, and non-local control sequence. We randomly selected 100 of the 2688 combinations, and for each combination we  aligned the sequences using Clustal \citep{Larkin:2007hz}, and estimated the underlying tree using PhyML and a HKY model \citep{Guindon:2003ba}. We favor a simple model, acknowledging potential improvements from more complex models and tuned parameters but noting they are unlikely  to substantially affect the results of our investigation.

\begin{center}
\begin{figure}
\includegraphics[trim = 3cm 2cm 0cm 2cm, scale=0.6]{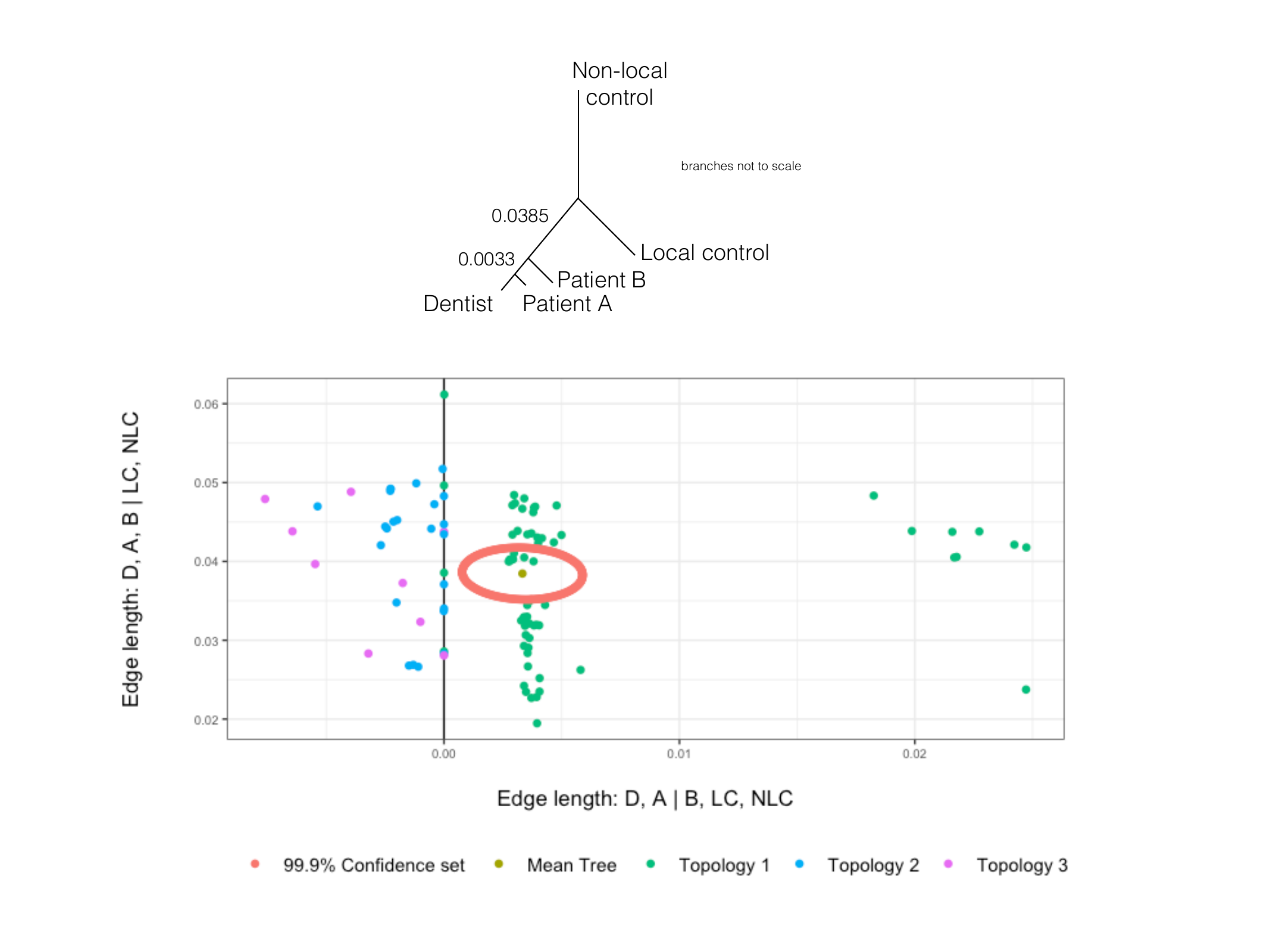}
\caption {The estimated phylogenies  of the HIV viruses of a dentist, two patients of the dentist, a control from the local population, and a control from a distinct population. The different phylogenies were obtained by permuting the representative sequences of each individual. Branch length units are substitutions per site. (top) The  Fr\'echet  mean of the 100 HIV phylogenies, which has Topology 1. (bottom) The MLMs of the 100 phylogenies with respect to the sample  Fr\'echet mean and the 99.9\% confidence set for the MLM of the true Fr\'echet mean. Points of the same color correspond to trees of the same topology. }
\label{hiv_figure}
\end{figure}
\end{center}

The projected trees under the sample MLM are depicted in Figure~\ref{hiv_figure}. The mean length of the branch separating the dentist and patients from the controls (Y-axis) is large relative to its variability. The null hypothesis that this edge is not present on the true tree is rejected with $p<10^{-10}$, and thus we conclude that the two patients contracted the virus from the dentist.  The remaining branch indicates the relative similarity of the dentist's sequences to those of patient A and patient B (which patient was infected closer to the date of blood sample collection), and we reject the null hypothesis that there is a leaf more closely related to the dentist than patient A ($p=10^{-4}$). These conclusions are consistent with \cite{ou92}. It is critical to note the simultaneous accounting for within and across patient similarity in this procedure, and thus it utilizes  the information available to the fullest extent and removes the need to separately consider intra/interperson sequence analysis as was necessary in the original investigation \citep{ou92}.

\subsection{Turtles:  Fr\'echet mean on a codimensional  stratum} \label{turtles}

\cite{Spinks:2013gq} investigated over-division of the {\it Pseudemys} ({\it P.})  genus of North American freshwater turtles, using a dataset of 86 turtles representing 13 taxa, of which 9 taxa represent subdivisions of the genus under question. 10 nuclear loci (6570 base pairs) and 3 mitochondrial genes (2209 base pairs) were used to build the trees, with full details regarding tree-building available in \citet[Section 2]{Spinks:2013gq}. A single turtle was chosen to represent each taxa and  the tree was built based on a single representative of each of the 13 taxa. This process was repeated 100 times to generate 100 trees, with each tree based on a different combination of turtle representatives of the taxa. The authors have kindly provided us with the 100 trees $T_1, \ldots, T_{100}$, which have 96 different topologies.

\begin{center}
\begin{figure}
\includegraphics[trim = 0cm 10cm 38cm 0cm, scale=0.75]{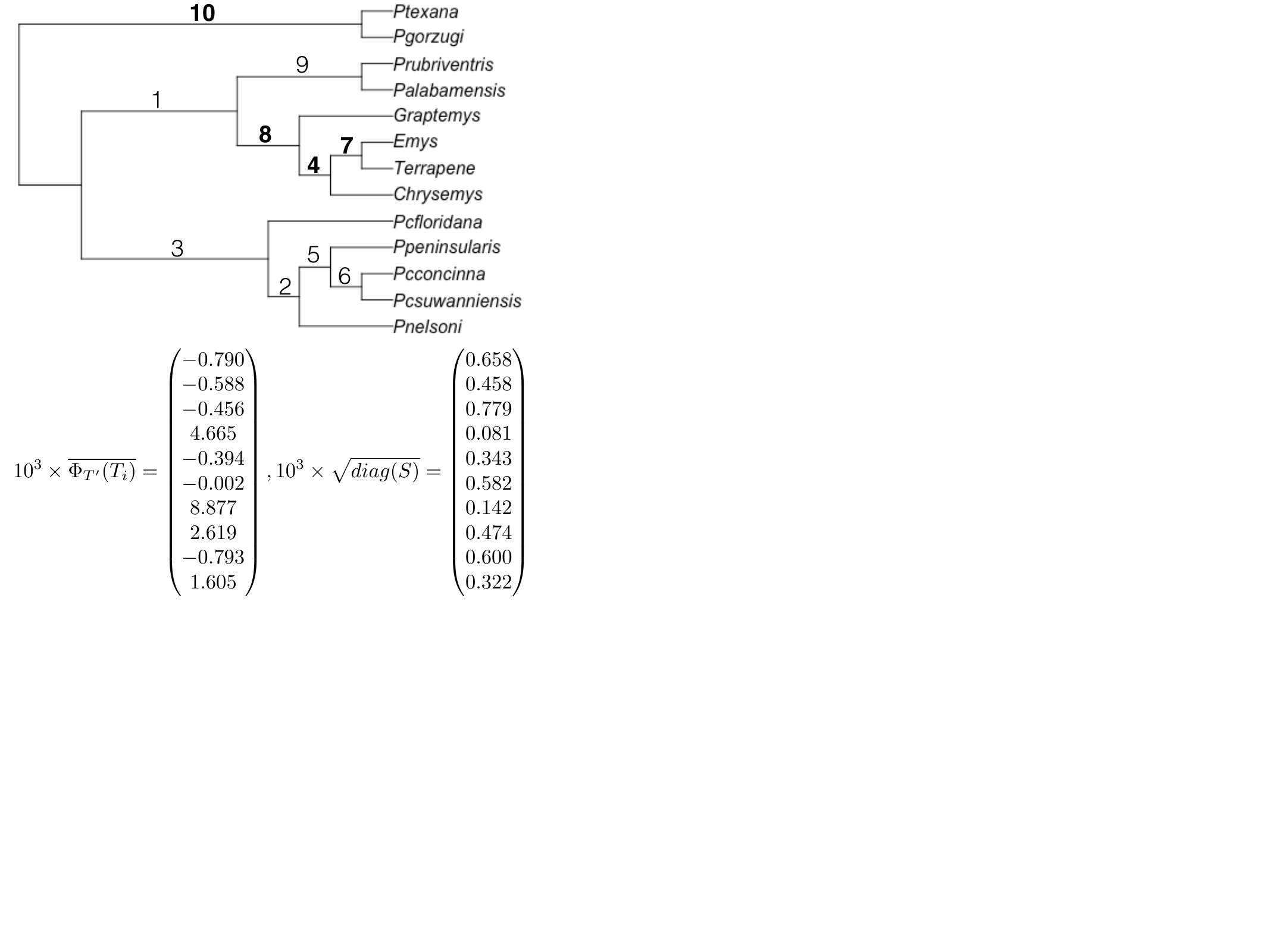}
\caption{The Fr\'echet mean of 100 turtle phylogenies is unresolved, and so a perturbation, $T'$ is shown (top, branch lengths not to scale). The rescaled sample mean and standard deviation of the log mapped trees $\Phi_{T'}(T_1), \ldots, \Phi_{T'}(T_{100})$ are also given (bottom). Each coordinate of $\Phi_{T'}(\cdot)$  corresponds to a branch on $T'$, and the branches are labelled on $T'$ in accordance with the order of the components of $\Phi_{T'}(\cdot)$. Branches 4, 7, 8 and 10 (in bold) correspond to MLM coordinates that are positive and large relative to their standard deviations, suggesting that these branches are strongly supported by $T_1, \ldots, T_{100}$. Branches 4, 7 and 8 correspond to outgroups (not shown).
}
\label{unresolved_figure}
\end{figure}
\end{center}

The sample mean of the 100 trees, which we call $\overline{T}$, falls on a stratum of codimension 4, and thus the formal inferential framework developed in Section \ref{procedure} cannot be applied. However, standard Euclidean multivariate data analysis tools can still be applied to the log mapped trees to assess the precision in estimating the branches. Define $T'$ to be the tree that is the weighted Fr\'echet mean of $\overline{T}$ with weight 0.95 and $T_1$ and weight 0.05. $T'$ is a tree close to $\overline{T}$ but in a top-dimensional stratum. Note that $T_1$ was chosen arbitrarily.

The tree $T'$ is shown in Figure \ref{unresolved_figure}, along with the mean ($\overline{\Phi_{T'}(T_i)}$) and standard deviation ($\sqrt{diag(S)}$) of $\Phi_{T'}(T_1), \ldots, \Phi_{T'}(T_{100})$. The branch labels on $T'$ correspond to the order of the coordinates of $\Phi_{T'}(\cdot)$. 4 branches have mean lengths that are more than 3 times larger than their standard deviations, indicating strong support for their presence on the true Fr\'echet mean tree (see histograms in Supplementary Appendix). Outgroups were introduced by \cite{Spinks:2013gq} to show relatives that are known to be taxonomically distinct, and 3 of the branches with strong support correspond to splits that separate the outgroups from the  {\it Pseudemys}  subgroups. The remaining branch of strong support corresponds to  {\it P. gorzugi}, corresponding identically with \citet[Table 3]{Spinks:2013gq}, which identifies this taxon to be taxonomically distinct.

The remaining 6 coordinates have negative means and large standard deviations, indicating that these branches are not consistently present on the trees $T_1, \ldots, T_{100}$. All other splits are highly questionable, and we concur with \cite{Spinks:2013gq} in concluding that the phylogenetic division of {\it Pseudemys} is oversplit, and the group should have  far fewer taxa than previously believed. This example illustrates that the log map is still a useful tool even when Fr\'echet means are not resolved, and can be used to quickly determine if a split is strongly or weakly supported by a collection of trees.

\section{Discussion} \label{discussion}


\subsection{Degeneracy} \label{boundary}
It is important to note that the asymptotics differ when the true Fr\'{e}chet mean falls on a stratum of codimension 1 \citep{blow2, barden2017}. The MLM remains multivariate normal on the branches whose means do not correspond to co-faces, with the co-facing branches converging to either a degenerate distribution or a truncated multivariate normal distribution \citep{blow2, hotz}.
Reduced variability in codimension asymptotics lead us to expect faster convergence to the degenerate branches, justifying our use of the usual asymptotics in the examples investigated in this manuscript.

\subsection{Extension to incorporate tree uncertainty} \label{uncertainty}
It is important to note that the trees that we consider as data points will generally be estimated, not known exactly \citep{holmes2003statistics}. Thus inherent in each observation is a possibly differing measure of certainty. We conjecture that this could be incorporated into the above procedure using the statistical framework of metaanalysis 
\citep{betta1}, 
and we are continuing research in this direction.

\subsection{Sources of tree-valued observations} \label{sources}
Our examples in this paper have exclusively focused on using within--species variability to more accurately reflect  variability in genetic data. However, many different processes give rise to phylogenetic tree--valued observations that could be used as inputs to this method. Gene trees, where each tree represents the phylogeny of a different location on the genome, provide another natural source of variability.
However, collections of gene trees may contain outlying trees,  and outliers should be removed before the confidence set is constructed. 

The method of \cite{bbb}, which was the first literature utilizing Fr\'{e}chet means to find consensus trees, used Fr\'{e}chet means and (scalar) Fr\'{e}chet variance to analyze trees generated by Bayesian MCMC draws. Using our multivariate extension above, it is possible to consider multiple directions of variability in this context.  

\subsection{Tree covariance} \label{covariance}
If tree--building information for $n$ different individuals from each of the $m$ taxa on the tree were be obtained, each individual from each taxa could be used exactly once to build $n$ independent trees. However, when differing numbers of individuals are obtained from each taxa, a choice must be made between discarding information (to equalize the number individuals from each group) or inducing dependence by repeating some individuals when building the trees. In Section \ref{examples}, we chose the latter option. This issue can be observed in Figures \ref{zika_figure}  and \ref{hiv_figure} via the clusters of the log mapped trees. Unfortunately, modeling dependence and incorporating it into the covariance estimate $\hat{\Sigma}$ is extremely challenging, because  the extent of dependence between two trees $T_i, T_j \in \tmm$ depends both  on the number of shared individuals used to build the trees, and also how closely  related these individuals are on the tree (a function of the unknown true tree $T^*$). We conjecture that ignoring this dependence is a second--order issue compared to  violations of identicality and  ignoring uncertainty in estimating the trees (Section \ref{uncertainty}), but investigation of this conjecture is an ongoing project.

\section{Concluding remarks} \label{conc}
The framework discussed here for representing collections of trees as points in Euclidean space opens phylogenetic tree analysis to many of the methods that have been developed for Euclidean space.
However, if the sample contains trees with conflicting topologies, the information regarding the topologies of the trees will be lost when the trees are considered as vectors (using  the log map).  For example, even if only 2 topologies are represented in the sample, the confidence set may span 3 or more orthants. For this reason, judgement should be exercised when deciding whether to analyze and visualize trees in their native space or under the log map in the more interpretable Euclidean space.

This method advances recent innovations with respect to describing tree-valued centre and variability, taking much inspiration from \cite{nye11} and \cite{bbb}.  We believe the most important progress made in this manuscript is the application of the statistical framework of variance {\it modeling}, rather than minimizing, to tree space. The proposal for using species replicates to generate collections of trees for summary and analysis may also prove fruitful by providing realistic measures of tree uncertainty. This is a known issue in phylogenetics  and we hope that the sampling method and the confidence set construction procedure described here contributes to a better understanding of both of these issues.

\section*{Software} The trees and scripts used to generated Table \ref{coverage_table}, as well as R  scripts for analyzing the data in  Section \ref{examples}, are available from github. R packages {\it ape, Rcmdr, scatterplot3d, rgl} and {\it mgcv} were used for the analysis and visualization, and we are grateful to the authors of these packages for their distribution.

\section*{Acknowledgements}
The author is indebted to Professor Tom  Nye of Newcastle University, who most kindly made available the code that underpins his PCA methods \citep{nye11,Nye:2014eh}, which assisted immensely in the implementation of the log map function. Many thanks to John Bunge for suggestions that significantly clarified the exposition;  Phil Spinks for the turtle trees of Section~\ref{turtles}; Megan Owen for many careful and constructive comments on an earlier draft; Giles Hooker and Marty Wells for funding that enabled completion of the project; and Sidney Resnick and Louis Billera for their support of the investigation and helpful suggestions at a formative stage. Two referees' thoughtful comments substantially improved every component of this manuscript and their time and care spent reviewing is very much appreciated.

\bibliographystyle{agsm}

\bibliography{tree-bibliography}
\end{document}